\documentclass[letterpaper,twocolumn,nofootinbib,aps,superscriptaddress]{revtex4}

\usepackage{amsmath}
\usepackage{mathrsfs}
\usepackage{amssymb}
\usepackage{hyphenat}
\usepackage{titlesec}
\usepackage{amsthm}
\usepackage{color}
\usepackage{bbm}
\usepackage{bm}
\usepackage{upgreek}
\usepackage{graphicx}
\usepackage[caption=false]{subfig}
\usepackage{pdftexcmds}
\usepackage{hyperref}
\usepackage[figure]{hypcap}

\newcounter{ifsup}

\DeclareFontFamily{U}{cbgreek}{}
\DeclareFontShape{U}{cbgreek}{m}{n}{
        <-6>    grmn0500
        <6-7>   grmn0600
        <7-8>   grmn0700
        <8-9>   grmn0800
        <9-10>  grmn0900
        <10-12> grmn1000
        <12-17> grmn1200
        <17->   grmn1728
      }{}
\DeclareFontShape{U}{cbgreek}{bx}{n}{
        <-6>    grxn0500
        <6-7>   grxn0600
        <7-8>   grxn0700
        <8-9>   grxn0800
        <9-10>  grxn0900
        <10-12> grxn1000
        <12-17> grxn1200
        <17->   grxn1728
      }{}

\makeatletter
\newcommand{\normalorbold}{%
  \ifnum\pdf@strcmp{\math@version}{bold}=\z@ bx\else m\fi
}
\makeatother

\newtheorem{theorem}{Theorem}
\newtheorem{athm}{Theorem}[section]

\newtheorem{propm}{Proposition}

\newtheorem{obs}[athm]{Observation}
\newtheorem{prop}[athm]{Proposition}
\newtheorem{alem}[athm]{Lemma}
\newtheorem{coro}[athm]{Corollary}

\newtheorem{definition}{Definition}

\newcommand*{\eins}{\ensuremath{\mathbbm 1}}
\def\gbm#1{{\let\lambda\uplambda \let\mu\upmu \let\rho\uprho \let\sigma\upsigma \let\tau\uptau \let\eta\upeta \bm{#1}}}

\newcommand*{\bbR}{\mathbb{R}}

\newcommand*{\cC}{\mathcal{C}}
\newcommand*{\cD}{\mathcal{D}}

\newcommand*{\cO}{\mathcal{O}}

\newcommand*{\cT}{\mathcal{T}}

\newcommand*{\sC}{\mathscr{C}}

\newcommand*{\sth}{\vartheta}

\newcommand*{\bA}{\mathbf{A}}
\newcommand*{\bB}{\mathbf{B}}
\newcommand*{\bC}{\mathbf{C}}
\newcommand*{\bM}{\mathbf{M}}

\newcommand*{\rB}{\mathrm{B}}

\newcommand*{\rd}{\mathrm{d}}
\newcommand*{\rI}{\mathrm{I}}

\newcommand*{\rS}{\mathrm{S}}

\newcommand*{\ket}[1]{\left|#1\right\rangle}
\newcommand*{\bra}[1]{\left\langle #1\right|}
\newcommand*{\proj}[1]{\ket{#1}\bra{#1}}
\newcommand*{\Tr}{\mathrm{Tr}}
\newcommand*{\fr}[2]{\frac{#1}{#2}}

\newcommand{\vect}[1]{\mathbf{#1}}

\newcommand{\be}{\begin{equation}}
\newcommand{\ee}{\end{equation}}
\newcommand{\n}{\textendash}
\newcommand{\m}{\textemdash}

\DeclareMathOperator*{\argmax}{arg\,max}

\hypersetup{
     bookmarks=false,         
    pdftoolbar=true,        
    pdfmenubar=true,        
    pdffitwindow=false,     
    pdfstartview={FitH},    
}

\begin{document}
\title{Quantifying memory capacity as a quantum thermodynamic resource}
\date{February 13, 2019}
\author{Varun Narasimhachar}
\email{nvarun@ntu.edu.sg}
\affiliation{Complexity Institute and School of Physical and Mathematical Sciences, Nanyang Technological University, 50 Nanyang Ave, Singapore 639798.}
\author{Jayne Thompson}
\affiliation{Centre for Quantum Technologies, National University of Singapore, Block S15, 3 Science Drive 2, 117543.}
\author{Jiajun Ma}
\affiliation{Center for Quantum Information, Institute for Interdisciplinary Information Sciences, Tsinghua University, 100084 Beijing, China.}
\author{Gilad Gour}
\affiliation{Institute for Quantum Science and Technology and Department of Mathematics and Statistics, University of Calgary, 2500 University Drive NW, Calgary, Alberta, Canada T2N 1N4}
\author{Mile Gu}
\email{mgu@quantumcomplexity.org}
\affiliation{Complexity Institute and School of Physical and Mathematical Sciences, Nanyang Technological University, 50 Nanyang Ave, Singapore 639798.}
\affiliation{Centre for Quantum Technologies, National University of Singapore, Block S15, 3 Science Drive 2, 117543.}

\begin{abstract}
The information\hyp carrying capacity of a memory is known to be a thermodynamic resource facilitating the conversion of heat to work. Szilard's engine explicates this connection through a toy example involving an energy\hyp degenerate two\hyp state memory. We devise a formalism to quantify the thermodynamic value of memory in general quantum systems with nontrivial energy landscapes. Calling this the \emph{thermal information capacity}, we show that it converges to the non\hyp equilibrium Helmholtz free energy in the thermodynamic limit. We compute the capacity exactly for a general two\hyp state (qubit) memory away from the thermodynamic limit, and find it to be distinct from known free energies. We outline an explicit memory\n bath coupling that can approximate the optimal qubit thermal information capacity arbitrarily well.
\end{abstract}

\maketitle

Szilard's adaptation of the \emph{Maxwell's demon} thought experiment, supplemented by Landauer's principle, illustrates a compelling connection between an entity's capacity to store information, on the one hand, and its capacity to deliver thermodynamic work, on the other~\cite{Szilard,LR14,Landauer,Bennett03}. In particular, the Szilard engine relies on storing information in an energy\hyp degenerate two\hyp state memory system. If such a memory is initialized in some pure state, one bit of information can be recorded onto it without expending any free energy. At the other extreme, if the memory started out in a maximally mixed state, no further information could be encoded onto it without first erasing its contents, which would entail tapping into an external free energy source.

This simple special case exemplifies a deeper connection between a memory's information capacity and its \emph{athermality} (i.e.\ departure from thermal equilibrium). How does this connection manifest in a general scenario where the memory is quantum mechanical, with internal states of differing energetic values? The athermality of such a memory may involve coherent superposition of energy eigenstates \cite{BCP14,XLF15,YZCM15,SSDBA15,WY16,CG16,CH16,MYGVG16,PJF16,ZSLF16}. In addition to the classical laws of thermodynamics, more general principles of nonequilibrium quantum thermodynamics \cite{SZAW03,SU09,KSdLU11,DRRV11,Reth,Nan,SSP14,BVMG15,PHS15,Sec,NU,GJBDM17} would then apply.
\begin{figure}[h]
    \includegraphics[width=\columnwidth]{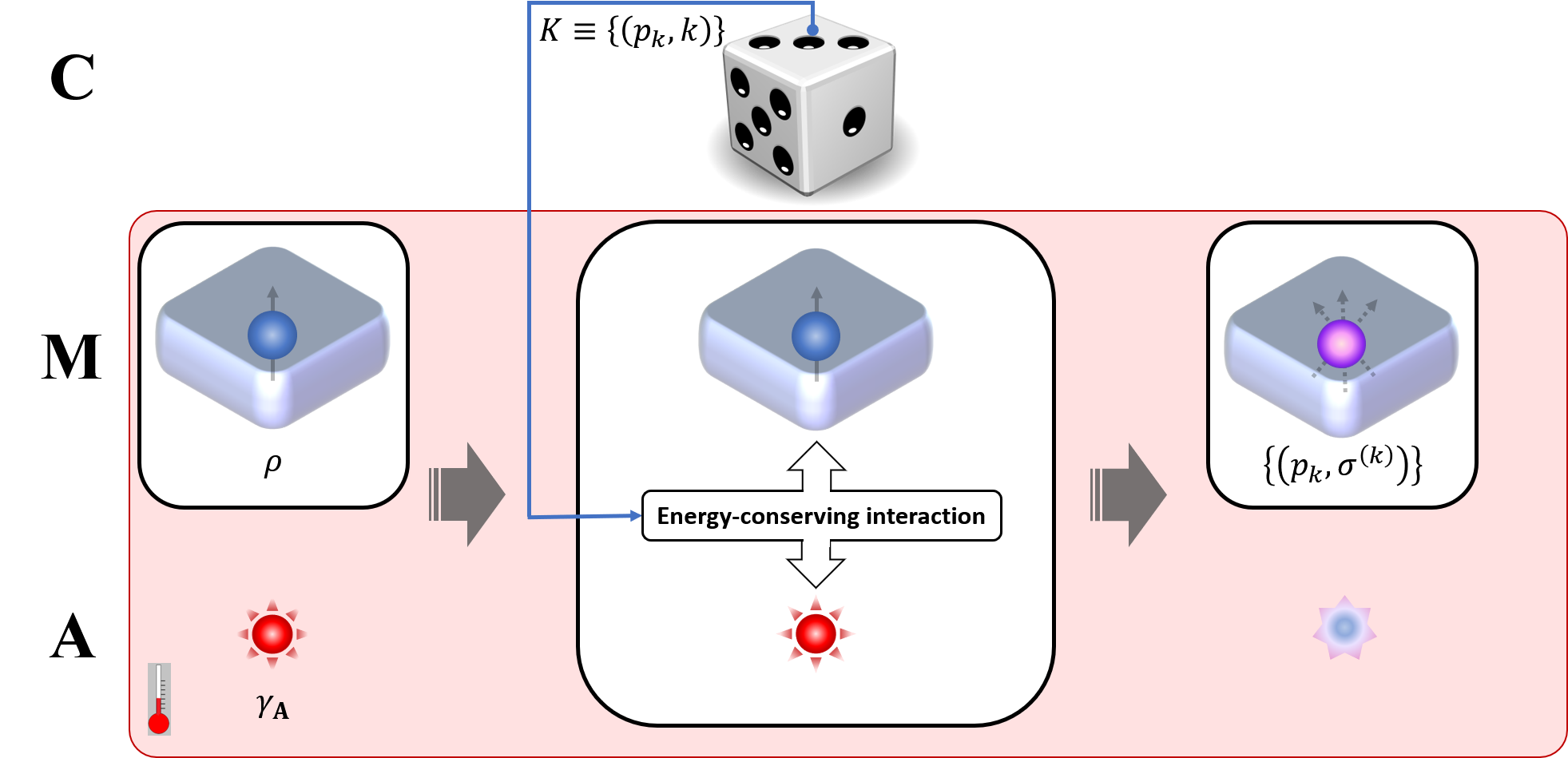}
    \caption{Thermally passive memory: the pertinent information variable $K\equiv\left\{\left(p_k,k\right)\right\}$ (contained in a classical system $\bC$) is recorded on the quantum memory $\bM$ through a $k$\hyp dependent energy\hyp conserving interaction of $\bM$ with an auxiliary system $\bA$ in its thermal state $\gamma_\bA$, transforming the memory's initial (``blank tape'') state $\rho$ to the ensemble $\{(p_k,\sigma^{(k)})\}$ of quantum state\n valued codewords. No free energy is used in this process, except that already present in $\rho$.}\label{figPM}
\end{figure}

In this Letter we formalize the thermodynamic value of memory capacity for quantum systems with general energy landscapes. To this end, we conceptualize a \emph{thermally passive memory}: writing onto such a memory is constrained to use no thermodynamic resource other than what the memory's initial state carries intrinsically. We define the resulting capacity as the given state's \emph{thermal information capacity}. We show that in the thermodynamic limit, this measure recovers the standard non\hyp equilibrium free energy. We also compute the thermal information capacity exactly for the case of a single two\hyp level (``qubit'') system away from the thermodynamic limit, establishing it as a distinct measure of athermality with operational relevance. We discuss a potential practical scheme to write onto a single\hyp qubit memory at a rate arbitrarily close to the capacity, finding a tradeoff between implementation speed and closeness of approximation.


\vspace{2ex}

\noindent\textbf{Framework.} We aim to capture the precise relationship between the thermodynamic inequilibrium (``athermality'') in arbitrary quantum states $\rho$ of a memory, and its storage capacity. To this end, we envision a \emph{thermally passive memory} (Fig.~\ref{figPM}): a quantum system $\bM$, initialized in state $\rho$, in a thermal environment of uniform temperature $T$. This effectively cuts off $\bM$ from any external sources of free energy, rendering it thermally passive.

Consider an arbitrary classical random variable  $K\equiv\left\{\left(p_k,k\right)\right\}$, representing information to be recorded onto $\bM$. This entails applying some $k$\hyp dependent operation on $\bM$ that transforms $\rho$ to a corresponding ``codeword'' state $\sigma^{(k)}$. Now, if $\rho$ were the state of thermal equilibrium, it would be impossible to take $\bM$ to any other state passively, making it useless as a memory. Any capacity for $\bM$ to passively record information, therefore, owes to the athermality of $\rho$.

Thermally passive encoding is formally captured by \emph{thermal operations} \cite{Nan}, which describe the possible state transformations of a system in contact with a single thermal bath. Left to equilibrate with the environment, $\bM$ would eventually reach its thermal, or Gibbs, state $\gamma\propto\exp\left(-H_\bM/k_\rB T\right)$, where $H_\bM$ is its free Hamiltonian. In this process of \emph{thermalization}, $\bM$ loses free energy and all other aspects of thermodynamic resourcefulness. A thermal operation is a more general type of resource\hyp depleting process, of which thermalization is a special case. It is an interaction of $\bM$ with a thermal auxiliary system $\bA$ (Fig.~\ref{figPM}): It starts with $\bM$ in some initial state $\rho$ uncorrelated with $\bA$ (which, by virtue of being thermal, is in its own local Gibbs state $\gamma_\bA$), followed by turning on an arbitrary energy\hyp conserving interaction between $\bM$ and $\bA$, and then decoupling the two again. The choice of system $\bA$ is left arbitrary, so long as it is prepared in the Gibbs state determined by its own Hamiltonian and the bath's temperature.

Passively encoding the classical variable $K$ on $\bM$ effectively transforms its state from $\rho$ to an ensemble $\cC\equiv\left\{\left(p_k,\sigma^{(k)}\right)\right\}$, where $\sigma^{(k)}=\cT^{(k)}(\rho)$ with $\cT^{(k)}$ some thermal operation for every $k$. The maximum amount of information that can be reliably recovered from $\cC$ by unrestricted readout is then given by its \emph{Holevo information}:
\be\label{defchi}
\chi\left(\cC\right)=S\left(\sum_kp_k\sigma^{(k)}\right)-\sum_kp_kS\left(\sigma^{(k)}\right),
\ee
where $S(\cdot)$ denotes the \emph{von Neumann entropy}. For a given initial state $\rho$, define $\sC(\rho)$ as the set of all codes $\cC$ consisting of codewords $\sigma^{(k)}$ accessible from $\rho$ by thermal operations\footnote{Equivalently, thermally passive encoding can be represented in terms of
  classical\n quantum (CQ) states of the classical variable and the memory,
  whereby $\sC (\rho )$ corresponds to the set of all CQ states accessible from
  $\rho $ under a generalized class of processes called \emph {conditioned
  thermal operations} \cite {NG17}. Since this more general framework is
  related but not essential to the present work, we shall state all of our
  results within the thermal operations framework.}. This set represents all possible ways that classical information can be written passively onto $\bM$, allowing arbitrary variations in the classical variable $K$ being written. Our main quantity of interest is the optimal amount of information that can be written in this way, given an initial resource state:
\begin{definition}[Thermal information capacity]\label{defIC}
The thermal information capacity (TIC) of the thermal memory $\bM$ initialized in blank state $\rho$ is defined as
\be
I_\mathrm{th}\left(\rho\right):=\sup\limits_{\cC\in\sC(\rho)}\chi\left[\cC\right].
\ee
\end{definition}

From the properties of the Holevo information, it follows that the TIC is always nonnegative. Another property of the TIC is that, as a function of the input state, it is strictly non\hyp increasing under thermal operations:
\be
I_\mathrm{th}\left(\cT[\rho]\right)\le I_\mathrm{th}\left(\rho\right)\quad\forall\rho,~\forall\mathrm{~thermal~operations~}\cT.
\ee
Thus, as expected, the TIC is a measure of thermodynamic resourcefulness of the state $\rho$, akin to free energy functions: a thermal operation acting on a given state can only result in a state with equal or lower TIC. It vanishes only when $\rho=\gamma$, and is positive otherwise.

Note that $I_\mathrm{th}$ is the \emph{absolute maximum} amount of information that we can encode within $\bM$, in a single shot, without energy expenditure. In particular, we assume no restriction on the operations required to \emph{decode} $K$ from $\sigma^{(k)}$, either to single\hyp shot processing or by energy considerations. Our primary motivation here is foundational: this allows us to study the efficacy of the writing process considered in isolation (e.g.\ the first step in a Szilard engine), and relate it to the initial athermality in $\bM$. Nevertheless, $I_\mathrm{th}$ has direct operational relevance in the context of remote probes operating in energy\hyp depleted environments. Such probes are constrained in their ability to harness free energy to store the information in their environment. The readout of this information may not need to be executed immediately, and may instead be deferred for more favourable conditions (e.g.\ after the probe has returned to a powered central facility). Examples of such settings arise in quantum sensing, where the operations that encode environmental data are generally thermal (e.g.\ unitary Hamiltonian evolution in metrology \cite{QMetro1,QMetro2}, or beamsplitter interactions with a thermal environment in the case of quantum illumination \cite{Lloyd08,Tan08,Bradshaw17}).\vspace{2ex}

\noindent\textbf{Thermodynamic limit.} A helpful point to start investigating the TIC is to consider its \emph{thermodynamic limit}, which concerns the average behaviour over a large number of independent, identically\hyp prepared (i.i.d.)\ instances. What is the optimal TIC per copy of a resource state $\rho$, in the thermodynamic limit? More precisely, this quantity is defined as
\be\label{ICiid}
I_{\mathrm{th}}^\infty\left(\rho\right):=\lim_{m\to\infty}\fr{I_\mathrm{th}\left(\rho^{\otimes m}\right)}m.
\ee
Apart from its own operational significance, the limiting i.i.d.\ value is useful as an upper bound on the single\hyp copy TIC. While the latter is in general difficult to compute, the i.i.d.\ limit can be calculated exactly using the theory of asymptotic equipartition, leading to the following result.
\begin{propm}\label{piid}
In the thermodynamic limit of infinitely many, independent and identically distributed (i.i.d.)\ copies, the optimal thermal information capacity per copy of a memory state $\rho$ is given by $I_{\mathrm{th}}^\infty=F(\rho)$, a quantum non\hyp equilibrium generalization of the Helmholtz free energy\footnote{The exact relationship between $F(\rho )$ and the Helmholtz free energy $A$ is
  given by $F(\rho )=\fr 1{k_\rB T}\left [A(\rho )-A(\gamma )\right ]$, where
  $\gamma $ is the Gibbs state. $F(\rho )$ equals the maximum expected work
  extracted in a thermal process acting on initial state $\rho $, quantified in
  units of ``work bits'', i.e.\ single energy\hyp degenerate qubit systems in
  pure states \cite {Nan,SSP14,Sec}.}, defined as the quantum relative entropy of $\rho$ with respect to the Gibbs state $\gamma$:
\be\label{GFE}
F(\rho):=S\left(\rho\|\gamma\right)\equiv\Tr\left(\rho\log_2\rho\right)-\Tr\left(\rho\log_2\gamma\right).
\ee
\end{propm}
We provide the proof in section~\ref{asym} of the Supplemental Material. Notably, the free energy emerges as the asymptotic TIC despite the readout's being unrestricted. This is because the optimal asymptotic code consists of pure eigenstates all equal in energy, and can therefore be read out by an energy\hyp conserving measurement. Our finding establishes that the TIC recovers  standard notions of free energy in the thermodynamic limit \cite{Sec,BG15}.

We now turn to the study of the TIC in the non\hyp i.i.d., or \emph{single\hyp shot}, regime. The science of general coherent thermal operations in this regime is nontrivial, but the special case of two\hyp level systems, or qubits, is relatively tractable.\vspace{2ex}
\begin{figure}[t]
    \includegraphics[width=\columnwidth]{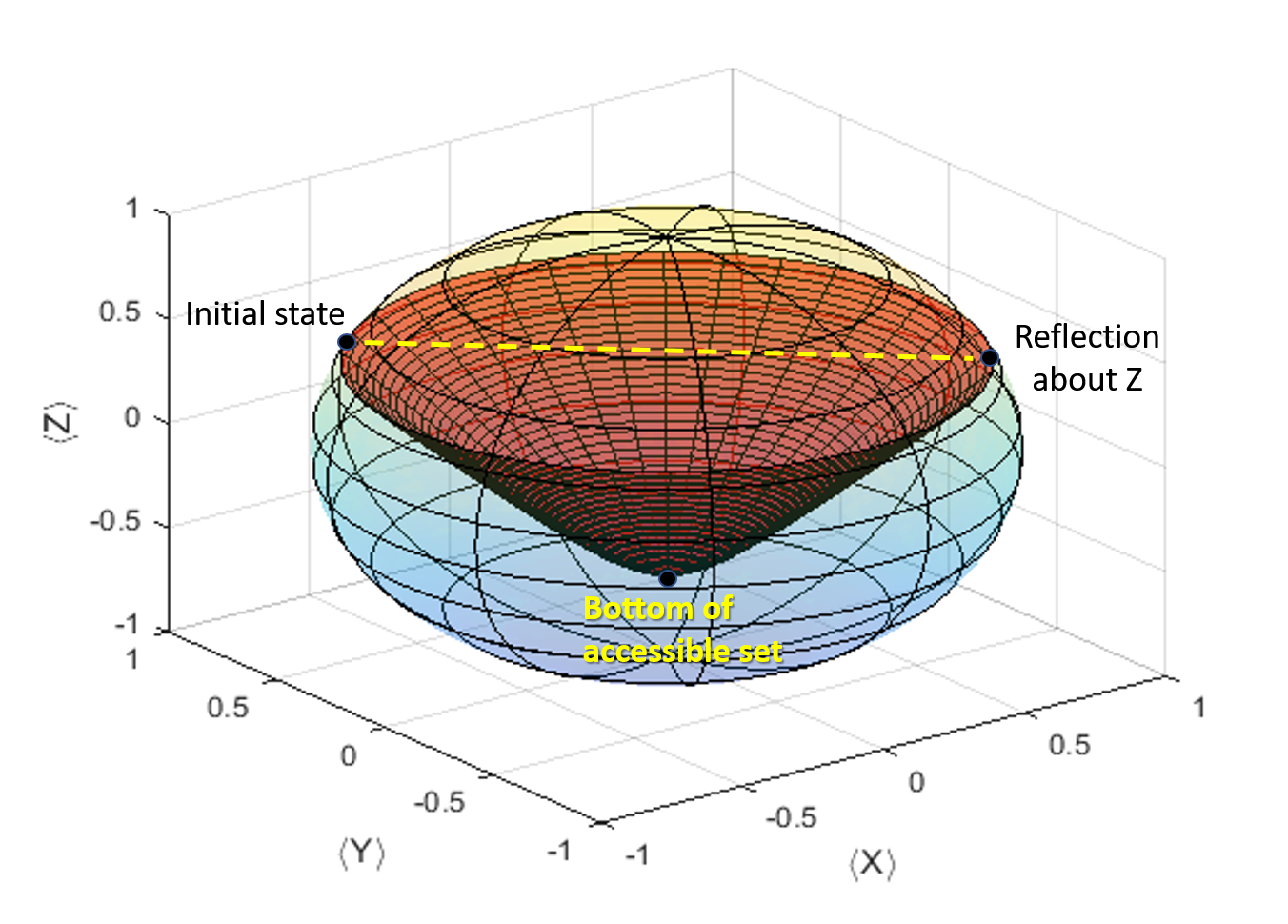}
    \caption{Bloch visualization of the set of states accessible by qubit thermal operations from a pure initial state; an informationally maximal code constructed from the accessible set comprises the three indicated extremal states as codewords.}\label{figacc}
\end{figure}

\begin{figure}[h]
  \centering
  \subfloat[$T=0$]{
    \includegraphics[width=.3\columnwidth]{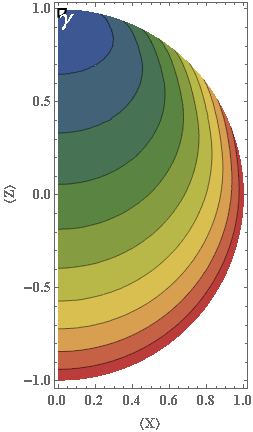}
    }
    \subfloat[$T=0.1\Delta E/k_{\mathrm B}$]{
    \includegraphics[width=.3\columnwidth]{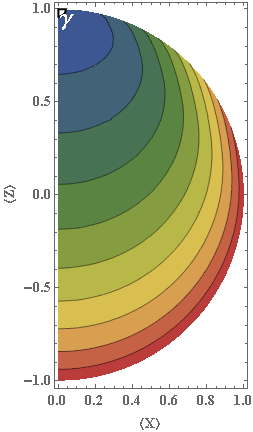}
    }
    \subfloat[$T=\Delta E/k_{\mathrm B}$]{
    \includegraphics[width=.3\columnwidth]{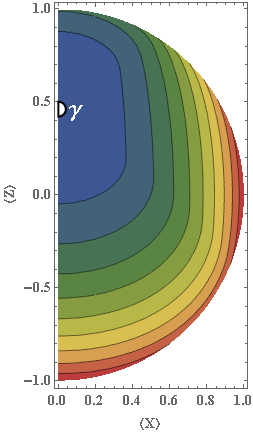}
    }\\
    \subfloat[$T=1.5\Delta E/k_{\mathrm B}$]{
    \includegraphics[width=.3\columnwidth]{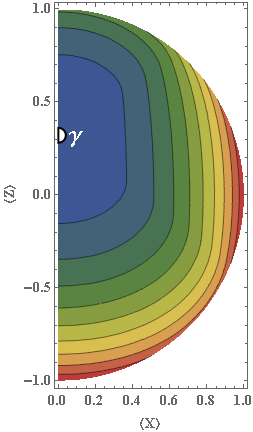}
    }
    \subfloat[$T=2\Delta E/k_{\mathrm B}$]{
    \includegraphics[width=.3\columnwidth]{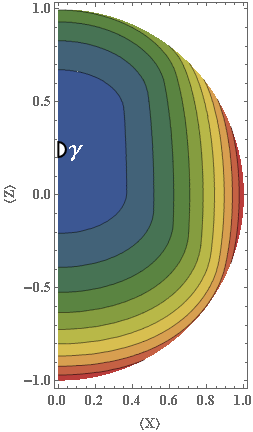}
    }
    \subfloat[$T\to\infty$\label{fig2d1}]{
    \includegraphics[width=.3\columnwidth]{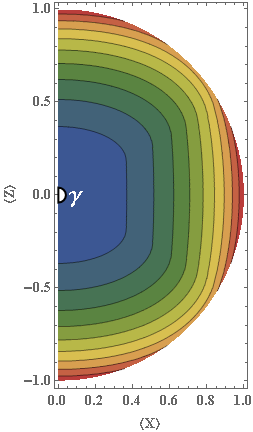}
    }\\
    \subfloat{
    \includegraphics[width=.5\columnwidth]{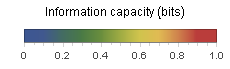}
    }
    \caption{Thermal information capacity (TIC) over different blank\hyp memory states in the $X^+Z$ section of the Bloch ball, for a qubit memory (with energy gap $\Delta E$) at various temperatures. The TIC of the Gibbs state $\gamma$ is zero, and is higher for states further away from $\gamma$. The zero\hyp temperature limit behaviour persists at temperatures as high as $0.1~\Delta E/k_\rB$; significant variation ensues in the $\cO\left(\Delta E/k_\rB\right)$ temperature range, while the high\hyp temperature limit resembles the information landscape of a non\n energy\hyp degenerate qubit memory.}\label{fig3d}
\end{figure}

\noindent\textbf{2\hyp level memory.} Consider a qubit memory $\bM$ governed by a (generally non\hyp degenerate) Hamiltonian $H_\bM=E_0\proj0+E_1\proj1$ and immersed in an ambient temperature $T$. Computing the TIC (Definition~\ref{defIC}) of a given initial state $\rho$ entails searching from the set $\sC(\rho)$ of codes accessible from $\rho$.
The concavity of the von Neumann entropy function implies that codes containing only extreme points of the accessible set will attain the optimum. This and other simplifications (detailed in section~\ref{appx} of the Supplemental Material) lead to our main result:
\begin{theorem}\label{thmain}
For a qubit memory $\bM$, an optimal code accessible thermally from an initial state $\rho$ is of the form
\be
\cC_q\equiv\left\{\left(\fr q2,\rho\right),\left(\fr q2,Z\rho Z\right),\left(1-q,\tilde\rho\right)\right\},
\ee
where $q\in[0,1]$, $Z=\proj0-\proj1$, and $\tilde\rho$ is the state at the tip of the accessible set (Fig.~\ref{figacc}). The thermal information capacity (TIC) of $\rho$ can then be determined by carrying out the single\hyp parameter optimization $I_{\mathrm{th}}(\rho)=\max\limits_{q\in[0,1]}\chi\left(\cC_q\right)$.
\end{theorem}
This optimization can be easily carried out numerically. Figure~\ref{fig3d} depicts the result: the TIC as a function of the initial state $\rho$, at various temperatures measured in relation to $\Delta E\equiv E_1-E_0$. The TIC understandably vanishes when $\rho$ equals the Gibbs state $\gamma$, and increases with \emph{athermality}, i.e.\ the departure of $\rho$ from this state. The Helmholtz free energy $F(\rho)$ [Eq.~\eqref{GFE}] is an operationally meaningful measure of athermality, and so we investigate the behaviour of $I_{\mathrm{th}}(\rho)$ in relation to $F(\rho)$ (Fig.~\ref{fig2d}). We see that the two resources vary similarly with $\rho$, but less so at lower temperatures. In section~\ref{appscat} of the Supplemental Material, we examine the TIC in relation with other resourcefulness measures, namely the purity and the relative entropy of coherence; we find the free energy to be better than these other resources as an indicator of the TIC. This is understandable, given the asymptotic convergence of the TIC to the free energy (Proposition~\ref{piid}).\vspace{2ex}

\begin{figure}[t]
    \subfloat[$T=0.1~\Delta E/k_{\mathrm B}$]{
    \includegraphics[width=.45\columnwidth]{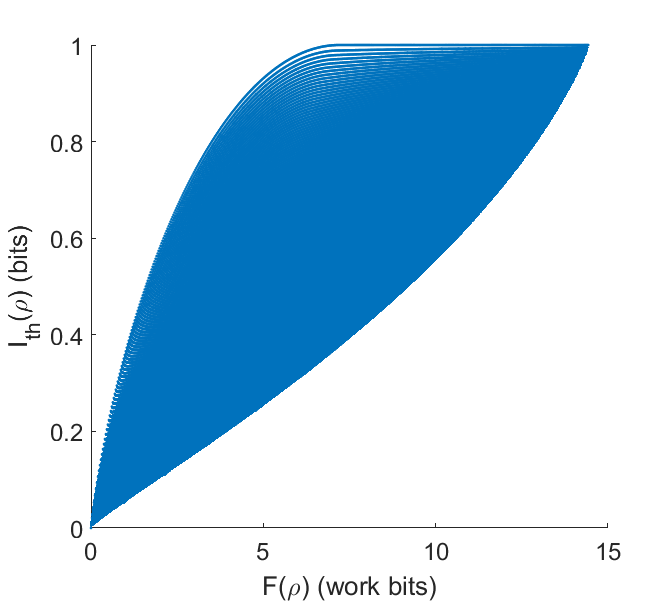}
    }
    \subfloat[$T=\Delta E/k_{\mathrm B}$]{
    \includegraphics[width=.45\columnwidth]{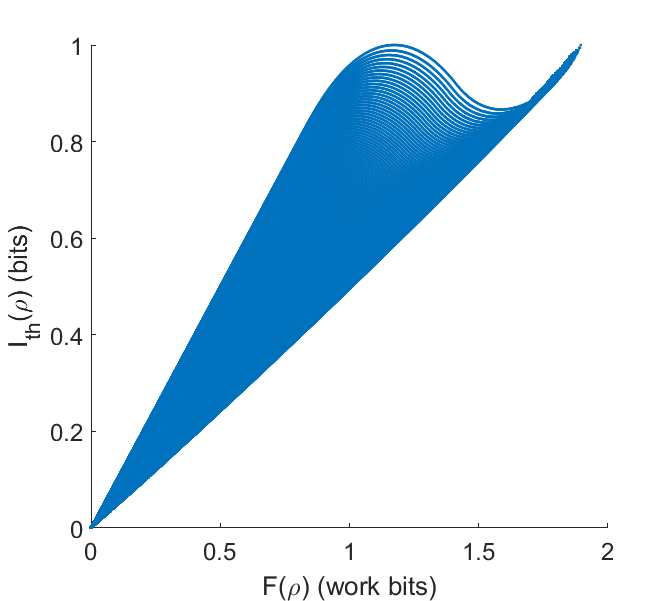}
    }\\
    \subfloat[$T=2~\Delta E/k_{\mathrm B}$]{
    \includegraphics[width=.45\columnwidth]{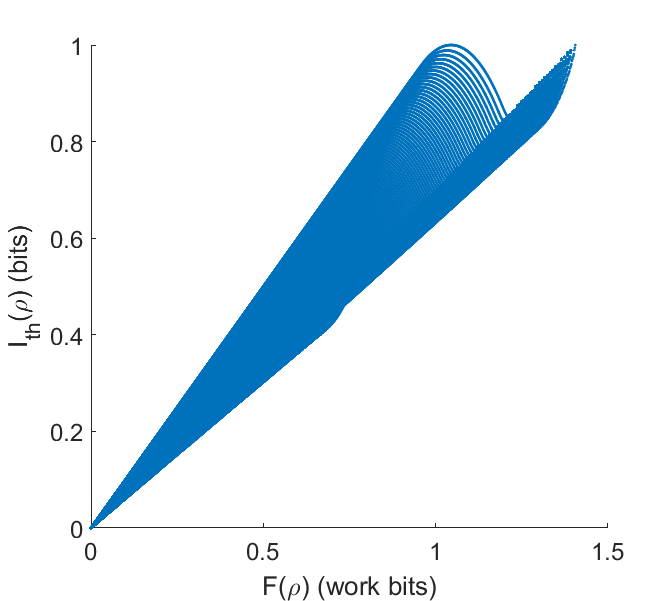}
    }
    \subfloat[$T\to\infty$]{
    \includegraphics[width=.45\columnwidth]{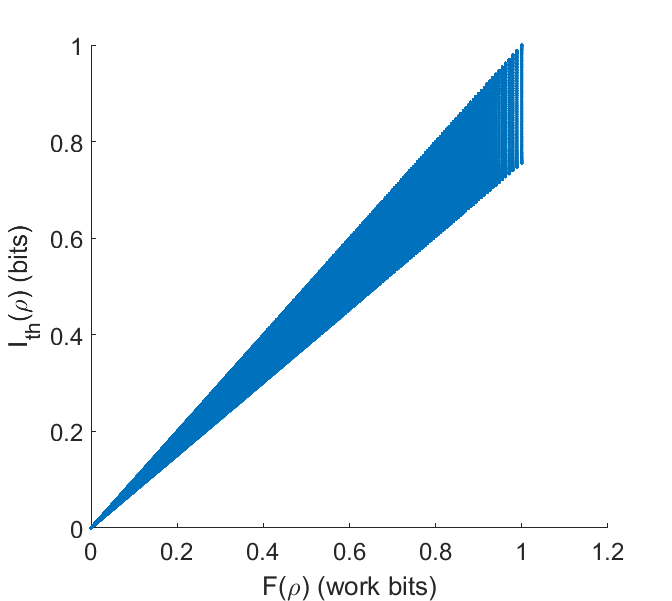}
    }
    \caption{Scatter plots of the thermal information capacity vs.\ non\hyp equilibrium Helmholtz free energy of qubit memory states: while the two resources are correlated in their state\hyp dependence, they are distinct, particularly at lower temperatures. In each plot, the top\hyp right point of maximum capacity corresponds to the initial state $\rho=\proj1$, the pure excited state. The maxima occurring to the left of this point correspond to initial states along the equator, e.g.\ $\rho=\proj+$. In the $T\to\infty$ limit, the two maximal regions get more and more similar in their free energy, as the latter converges to the purity (or ``negentropy'') of $\rho$.}\label{fig2d}
\end{figure}

\noindent\textbf{Towards implementation.} The thermal operations framework, which we have used to model the encoding process, is agnostic about the existence of a practically feasible auxiliary system $\bA$ and coupling to realize a desired thermal operation (see \cite{NYH17} for a detailed discussion). Thus, we would like to go beyond the abstraction of thermal operations and construct a concrete realization. To this end, we now probe an interaction of the qubit memory $\bM$ with a bosonic mode bath tuned to $\bM$'s energy gap, interacting with the latter via a Jaynes\n Cummings coupling.

We refer again to Fig.~\ref{figacc} showing the three states constituting an optimal code obtainable from a given initial state. The initial state itself being one of these, another results from reflecting the initial state about the Pauli Z axis, while the third lies at the tip of the convex cone of accessible states. Reflection about Z is represented by the unitary transformation $Z$, which can be effected simply by evolving the memory system under its free Hamiltonian for a suitable length of time. Transforming to the third codeword state, however, requires population inversion relative to the initial state, which cannot be achieved perfectly by a Jaynes\n Cummings coupling owing to asynchronicity between the Rabi oscillations within different memory\n bath energy levels. Nevertheless, we found that the optimal capacity can be approximated arbitrarily well, albeit at the cost of longer running time (Fig.~\ref{figimp}): this mirrors the power\n efficiency tradeoff in the performance of heat engines. The phase transition\n like jumps occur due to the above\hyp mentioned Rabi oscillations whose collective effect on the qubit's marginal state is irregular in time. The degree of population inversion required to meet a given efficiency is generally achieved at similar times over short ranges of temperature, but at certain critical temperatures where it just begins to fail, the irregular time\hyp dependence of the population inversion leads to a long period of oscillations where this failure persists, until a sufficient inversion level is finally reached around a different time regime. This new inversion level again remains sufficient to meet the required efficiency, until the next critical temperature is hit, and so on. The \emph{downward} dip of some of the curves with increasing temperature seems counterintuitive. We conjecture that this is a consequence of the fall in optimal capacity with increasing temperature, thus rendering it easier to approach. Technical details about these results are provided in section~\ref{appimp} of the Supplemental Material.\vspace{2ex}

\noindent\textbf{Discussion.} We probed the thermodynamical limitations of the capacity of a quantum system to store information. We defined a \emph{thermally passive quantum memory} as one which is written onto without access to free energy sources, and \emph{thermal information capacity} as the capacity of such a memory. After determining that the thermal information capacity approaches the non\hyp equilibrium free energy in the thermodynamic limit, we computed it away from the thermodynamic limit for a single\hyp qubit memory, showing it to be distinct from known free energies. We then described a proposal for approximating the optimal encoding strategy through a Jaynes\n Cummings interaction of the memory with a Bosonic bath.

\begin{figure}[t!]
    \includegraphics[width=\columnwidth]{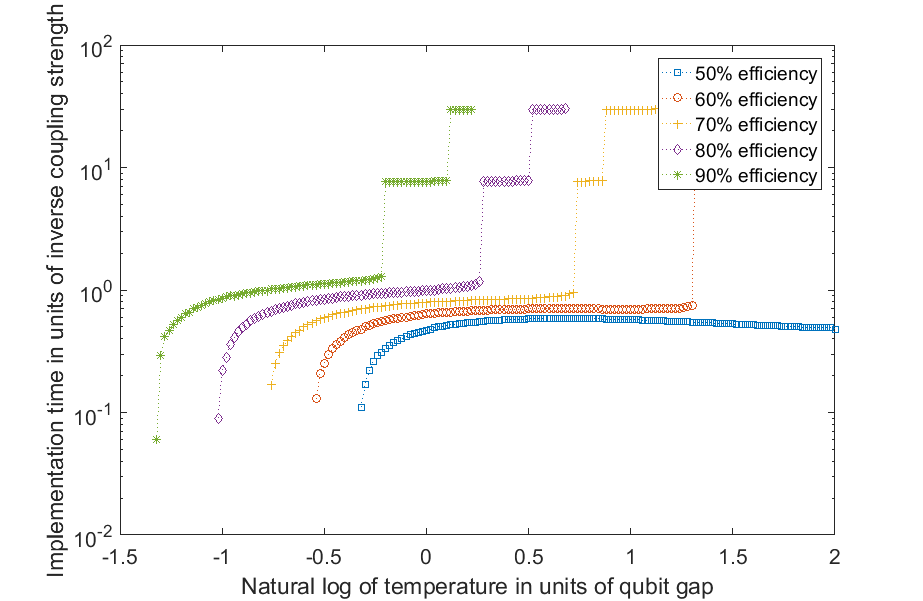}
    \caption{Time taken by a Jaynes\n Cummings coupling to approximate the optimal qubit thermal information capacity to various efficiencies, vs.\ bath temperature. The speed\hyp efficiency tradeoff is reminiscent of a heat engine's performance.}\label{figimp}
\end{figure}

The connection between information processing and thermodynamics in general quantum settings has many facets. Among these, one that has stimulated significant interest is understanding the role of quantum effects, such as coherence, in thermodynamic contexts~\cite{CST,Coh,NG15,catcoh,KLOJ15,KA16,SU09,nege,PKSK13,BMC16,FGPP17}. Our work provides a new perspective, in investigating how quantum coherence improves our capacity to store information without additional free energy. Furthermore, passive information storage has operational relevance in sensing applications where the means through which a probe encodes environmental information is implicitly energy\hyp conserving~\cite{Tan08,Bradshaw17,QMetro1,QMetro2}. The thermal information capacity then provides an ultimate upper bound in how much information such a probe can store.
%

Another natural question that follows from thermally passive encoding is: How can we use the encoded information in a way that is also subject to thermodynamic constraints? The primary challenge here is that the direct way to read out encoded information\m measurement\m lies outside the thermal operations framework, as measurement\hyp induced collapse can drive systems out of thermal equilibrium. Hence, we need to explore more sophisticated frameworks, such as thermally passive coupling between the memory and the system it is storing information about. Developments in such directions could enable a full description of generalized quantum Szilard engines that take full advantage of initial memory states that can exist in quantum superpositions of non\hyp degenerate energy eigenstates.

\vspace{2ex}

\noindent\textbf{Acknowledgments.} The authors thank Syed Assad, Diane Donovan, Ping Koy Lam, and Bevan Thompson for helpful discussions. The work is supported by the National Research Foundation of Singapore (NRF Fellowship Reference No.\ NRF NRFF2016-02),  the John Templeton Foundation (Grant No.\ 54914), the FQXi Large Grant ``Observer\hyp Dependent Complexity: The Quantum\n Classical Divergence over `What is Complex?{'}'', the National Research Foundation and L'Agence Nationale de la Recherche joint project NRF2017-NRFANR004 VanQuTe and the Singapore Ministry of Education Tier 1 RG190/17. GG acknowledges support from the Natural Sciences and Engineering Research Council of Canada.

\clearpage
\onecolumngrid

\begin{center}
\textbf{\large Supplemental Material}
\end{center}



\renewcommand{\thepage}{S\arabic{page}}
\renewcommand\thesection{S\arabic{section}}
\setcounter{page}{1}
\setcounter{section}{0}
\makeatletter
\numberwithin{equation}{section}
\renewcommand{\theequation}{S\arabic{section}.\arabic{equation}}
\addtocounter{ifsup}{1}
\numberwithin{figure}{ifsup}
\renewcommand{\thefigure}{S\arabic{figure}}
\setcounter{figure}{0}

\section{TIC in the thermodynamic limit}\label{asym}
We are interested in determining
\be\label{AICiid}
I_{\mathrm{th}}^\infty\left(\rho\right):=\lim_{m\to\infty}\fr{I_{\mathrm{th}}\left(\rho^{\otimes m}\right)}m,
\ee
under the class of thermal operations (TO) on infinitely many copies of a $d$\hyp dimensional elementary system with Hamiltonian $H$, with the associated Gibbs state $\gamma$. For convenience, we assume $H$ has no degeneracy; our arguments can be easily generalized to degenerate cases.

The result of \cite{SReth} states that, given two resources $\rho$ and $\sigma$, the conversion $\rho^{\otimes m}\otimes\gamma^{\otimes n_m}\mapsto\sigma^{\otimes n_m}\otimes\gamma^{\otimes m}$ in the limit $m\to\infty$ is possible under TO (allowing a conversion error that vanishes in the limit) at the optimal rate
\be\label{iidrate}
\limsup_{m\to\infty}\fr{n_m}{m}=\fr{F\left(\rho\right)}{F\left(\sigma\right)},
\ee
where $F(\rho):=S(\rho\|\gamma)=\Tr\left(\rho\log_2\rho\right)-\Tr\left(\rho\log_2\gamma\right)$.

We first convert the given $m$ copies of the general resource $\rho$ to some standard resources with the same amount of free energy; the asymptotic reversibility mentioned above ensures that the TIC of these standard resources\m which happens to be easier to calculate\m is equal to that of the general ones.

The standard resources of our choice are pure states of the form
\be
\Psi\left(\vect j\right)\equiv\bigotimes_{k=0}^{n-1}\proj{E_{j_k}},
\ee
where $\vect j$ is a collection of (an as\hyp yet\hyp unspecified number) $n$ indices, each chosen from $\{0,1\dots,d-1\}$. The energy of this state is given by
\be
E(\vect j)=\sum_{k=0}^{n-1}E_{j_k}.
\ee
The number $n$ is expected to be very large, while the possible values for each $j_k$ number $d$. Therefore, $\vect j$ will typically have repeating indices. Define the vector of frequencies, $\vect f$, by
\be
f_j:=\left|\left\{k\in\{0,1\dots,n-1\}|j_k=j\right\}\right|.
\ee
Then, the rank of the degenerate subspace of energy $E(\vect j)$ is given by the multinomial coefficient
\be
\mu(\vect f):=\left(\begin{array}{c}
n\\
f_0,f_1\dots,f_{d-1}
\end{array}\right).
\ee
Arbitrary unitaries within this subspace are energy\hyp conserving, and therefore TO. Thus, starting from $\Psi(\vect j)$ (or any other pure state in this subspace), we can use TO to construct an ensemble of $\mu(\vect f)$ equally\hyp probable orthonormal pure states, which achieves a Holevo rate of $\log_2\mu(\vect f)$.

We now draw inspiration from the theory of asymptotic equipartition to determine the best choice of $\vect f$. As an Ansatz, let us fix $n$ and set $f_j=g_jn$. Using Eq.~(\ref{iidrate}), we find the number of initial copies of $\rho$ required for constructing $\Psi(\vect j)$:
\be
m=\fr{\sum_jf_jS\left(\proj{E_j}\|\gamma\right)}{S\left(\rho\|\gamma\right)}=n\sum_j\fr{g_j\log_2\left(Zg_j^{-1}\right)}{S\left(\rho\|\gamma\right)},
\ee
where $Z=\sum_jg_j$ is the single\hyp system partition function. Using this construction, we can lower\hyp bound the asymptotic TIC rate defined in Eq.~\eqref{AICiid}:
\be\label{i.i.d.bound}
I_{\mathrm{th}}^\infty\left(\rho\right)\ge\lim_{m\to\infty}\fr{\log_2\mu(\vect f)}m=S\left(\rho\|\gamma\right)\equiv F(\rho).
\ee
To see that this is also an upper bound, we note the following. The final memory state contains correlations with the classical variable, which is itself a thermodynamic resource that can be capitalized to recover copies of the original resource $\rho$ at precisely the rate $F(\rho)$. If this were not also an upper bound, more of the initial resource could be reconstructed than we began with, thereby leading to a net creation of resource under TO. Since this is forbidden, the bound works both ways, establishing Proposition~\ref{piid} of the main text.

\section{Technical results for qubit TIC}\label{appx}
Here we provide the technical results used in proving Theorem~\ref{thmain} about TIC under qubit TO. We adopt a convenient shorthand, denoting a general state of the qubit memory $\bM$ by
\be
\eta[r,\alpha]:=\left(\begin{array}{lr}
r&\alpha\\
\alpha^*&1-r
\end{array}\right),
\ee
where the matrix representation is relative to the energy basis $\left\{\ket0,\ket1\right\}$. The Gibbs state is given by $\gamma=\eta[g,0]$, where $g=\exp(-\beta E_0)/\left[\exp(-\beta E_0)+\exp(-\beta E_1)\right]$. We also define
\be\label{deflamb}
\lambda:=\fr{1-g}g=\exp\left[\beta(E_0-E_1)\right].
\ee
Our aim is to compute the TIC, defined as
\be\label{ATOIC}
I_{\mathrm{th}}\left(\rho\right)=\sup_{\cC\in\sC(\rho)}\chi(\cC),
\ee
where $\sC(\rho)$ is the set of all codes constructed from codewords contained in the set [call it $\sth(\rho)$] of states accessible by TO from the initial state $\rho$. The results of Ref.~\cite{SNan} imply that
\be
\sth(\rho)=\left\{\eta[s,\beta]:s\in\left[r,1-\lambda r\right]
\left|\beta\right|\le\kappa_s
\right\},
\ee
with
\be
\kappa_s:=\left|\alpha\right|\fr{\sqrt{\left[\lambda s+r-1\right]\left[\lambda r+s-1\right]}}{\left|(\lambda+1)r-1\right|}.
\ee
\begin{obs}
The optimization in Eq.~(\ref{ATOIC}) can be restricted to codes $\cC$ containing only the extreme points $\eta[s,\kappa_se^{i\phi}]$ of $\sth(\rho)$. For, if some code $\tilde\cC$ contains the codeword $\left(p,q\sigma_1+[1-q]\sigma_2\right)$ with $p>0$; $0<q<1$; and $\sigma_1\ne\sigma_2$ both in $\sth(\rho)$, we can construct another code $\cC$, identical to $\tilde\cC$ except with this codeword replaced by two others, namely $\left(pq,\sigma_1\right)$ and $\left(p[1-q],\sigma_2\right)$. By the concavity of the von Neumann entropy, $\chi(\cC)>\chi\left(\tilde\cC\right)$.
\end{obs}
\begin{alem}
The optimization in Eq.~(\ref{ATOIC}) can be further restricted, to codes consisting solely of pairs $\left(\eta[s,\kappa_se^{i\phi}],\eta[s,-\kappa_se^{i\phi}]\right)$ of extremal states lying on opposite sides of the Z axis in the Bloch representation, with both states in a pair of a given $s$ occurring with equal probability. Explicitly, such a code takes the form
\be\label{paircode}
\cC=\left\{\left(\fr{p_j}2,\eta\left[s_j,\pm\kappa_{s_j}e^{i\phi_j}\right]\right)\right\}.
\ee
\end{alem}
\begin{proof}
For some code $\tilde\cC=\left\{\left(p_j,\eta\left[s_j,\beta_j\right]\right)\right\}$ constructed from extreme points of $\sth_\rS(\rho)$, the Holevo quantity is given by
\be
\chi\left(\tilde\cC\right)=S\left(\sum_jp_j\eta\left[s_j,\beta_j\right]\right)-\sum_jp_jS\left(\eta\left[s_j,\beta_j\right]\right).
\ee
Now, we will show that the Holevo rate of the corresponding paired code $\cC$, as in Eq.~\eqref{paircode}, is no smaller than $\chi\left(\tilde\cC\right)$.
\begin{align}
\chi(\cC)
&=S\left(\sum_j\fr{p_j}2\left(\eta\left[s_j,\beta_j\right]+\eta\left[s_j,-\beta_j\right]\right)\right)-\sum_j\fr{p_j}2\left[S\left(\eta\left[s_j,\beta_j\right]\right)+S\left(\eta\left[s_j,-\beta_j\right]\right)\right]\nonumber\\
&=S\left(\sum_j\fr{p_j}2\left(\eta\left[s_j,\beta_j\right]+\eta\left[s_j,-\beta_j\right]\right)\right)-\sum_jp_jS\left(\eta\left[s_j,\beta_j\right]\right)\nonumber\\
&\ge \fr12\left[S\left(\sum_jp_j\eta\left[s_j,\beta_j\right]\right)+S\left(\sum_jp_j\eta\left[s_j,-\beta_j\right]\right)\right]-\sum_jp_jS\left(\eta\left[s_j,\beta_j\right]\right)\nonumber\\
&=S\left(\sum_jp_j\eta\left[s_j,\beta_j\right]\right)-\sum_jp_jS\left(\eta\left[s_j,\beta_j\right]\right)=\chi\left(\tilde\cC\right).
\end{align}
The second line follows from the unitary relation (namely, through the unitary $Z$) between the states within each pair, the third from the concavity of the von Neumann entropy, and the fourth from the existence of a common unitary (again, $Z$) connecting corresponding codewords in the two half\hyp codes.

Together with the previous observation, this implies that an extremal code of the paired form \eqref{paircode} will attain the optimum TIC.
\end{proof}
We now note that $S\left(\eta[s,\kappa_se^{i\phi}]\right)$ is independent of $\phi$; it is effectively a function of $s$, which we denote $S(s)$. The Holevo information of a paired code such as in Eq.~(\ref{paircode}) is given by
\begin{align}
\chi(\cC)&=S\left[\sum_jp_j\left(\fr{\eta[s_j,\kappa_{s_j}e^{i\phi_j}]+\eta[s_j,-\kappa_{s_j}e^{i\phi_j}]}2\right)\right]-\sum_jp_j\left[\fr{S\left(\eta[s_j,\kappa_{s_j}e^{i\phi_j}]\right)+S\left(\eta[s_j,-\kappa_{s_j}e^{i\phi_j}]\right)}2\right]\nonumber\\
&=S\left[\sum_jp_j\eta[s_j,0]\right]-\sum_jp_j\left[\fr{S\left(\eta[s_j,\kappa_{s_j}e^{i\phi_j}]\right)+S\left(\eta[s_j,-\kappa_{s_j}e^{i\phi_j}]\right)}2\right]\nonumber\\
&=h\left(\sum_jp_js_j\right)-\sum_jp_jS(s_j)=:h\left(\bar s\right)-\sum_jp_jS(s_j),
\end{align}
where we recall that $h(\cdot)$ denotes the binary entropy function. We can now state the optimization in Eq.~(\ref{ATOIC}) as
\be
I_\mathrm{th}\left(\rho\right)=\max_{\bar s\in[r,1-\lambda r]}\left[h(\bar s)-\xi(\bar s)\right],
\ee
where
\be\label{defxi}
\xi(\bar s):=\min_{(\vect p,\vect s)|\sum_jp_js_j=\bar s}\sum_jp_jS(s_j).
\ee
Here it is to be understood that $\vect p$ is a probability distribution and that the $s_j$ are constrained to lie in $[r,1-\lambda r]$.
\begin{prop}\label{Scon}
$S(s)$ is concave for $s\in[r,1-\lambda r]$.
\end{prop}
We will prove this proposition through several steps. We will largely exploit the simplicity of the qubit case, wherein all spectral properties of a density operator reduce to functions of a single parameter. In particular, consider the von Neumann entropy $S(\sigma)$, introduced already, and the determinant, which we shall denote $D(\sigma)$. They can both be expressed in terms of a single parameter. One possible choice for this parameter is the smaller of the two eigenvalues of $\sigma$, which we here denote $t$; note that $t\in[0,1/2]$. As a function of $t$, the von Neumann entropy and determinant are
\begin{align}
S(t)&=h(t)\equiv-t\log_2t-(1-t)\log_2(1-t);\nonumber\\
D(t)&=t(1-t).
\end{align}
We note that our use of the symbols $S$ and $D$ here is to refer not to specific functional forms, but rather to the von Neumann entropy and the determinant treated as variables. When one of these symbols is followed by an argument, it is then (and only then) intended to convey the behaviour of the variable as a function of the said argument. In particular, this means that the following functional forms are all distinct, even though they all represent the von Neumann entropy:
\begin{enumerate}
\item$S(s)$ as a function of $s$ (which parametrizes the special family $\eta[s,\kappa_se^{i\phi}]$ of density operators);
\item$S(t)$ as a function of $t$ (the smaller eigenvalue of the density operator);
\item$S(\sigma)$ as a function of $\sigma$ (the density operator);
\item$S(D)$ as a function of $D$ (the determinant of the density operator).
\end{enumerate}
\begin{alem}\label{DS}
Over qubit density operators, the von Neumann entropy $S$ is an invertible function of the determinant $D$; specifically, $S$ is a strictly increasing, concave function of $D$.
\end{alem}
\begin{proof}
Recall that both these quantities are effectively functions of the single parameter $t\in[0,1/2]$. Both functions are well\hyp defined and continuous in the interior of this region; $S(t=0)$ can be set to $0$ using the limit as $t\to0^+$. Denoting the total derivative with respect to $t$ by an overhead dot,
\begin{align}
\dot S(t)&=\log_2\fr{1-t}t;\nonumber\\
\dot D(t)&=1-2t.\label{SDt}
\end{align}
These are both well\hyp defined and strictly positive in the interior of the parametric region. In other words, both $S(t)$ and $D(t)$ are well\hyp defined and strictly increasing in the region. Therefore, $S$ is a strictly increasing function of $D$. It remains to show the function's concavity.

First, using Eqs.~\eqref{SDt}, we have
\begin{align}
\ddot S(t)&=-\fr{\log_2e}{t(1-t)};\nonumber\\
\ddot D(t)&=-2.\label{ddt}
\end{align}
Also from Eqs.~\eqref{SDt},
\be
\left.\fr{\rd S}{\rd D}\right|_t=\fr{\dot S(t)}{\dot D(t)}=\fr{\log_2\fr{1-t}t}{1-2t},
\ee
admitting the definition
\be
\left.\fr{\rd S}{\rd D}\right|_{t=1/2}=2\log_2e
\ee
through L'H\^opital's rule and Eqs.~\eqref{ddt}. Moving on,
\be
\left.\fr{\rd^2S}{\rd D^2}\right|_t=\fr1{\dot D(t)}\left.\fr\rd{\rd t}\left(\fr{\rd S}{\rd D}\right)\right|_t=\left(-\log_2e\right)\fr{1-2t-2t(1-t)\ln\fr{1-t}t}{t(1-t)(1-2t)^3}.\label{ddSD}
\ee
The denominator is nonnegative for $t\in[0,1/2]$. Now let
\be
f(t):=1-2t-2t(1-t)\ln\fr{1-t}t.
\ee
Evidently, $f(1/2)=0$, while one may show easily (e.g., using L'H\^opital's rule), that
\be
\lim\limits_{t\to0^+}f(t)=1.
\ee
The function is well\hyp defined and smooth in the interior of the region, where
\be
\dot f(t)=-2(1-2t)\ln\fr{1-t}t>0.
\ee
Concavity of $S(D)$ follows from Eq.~\eqref{ddSD}.
\end{proof}
Note that this functional relationship holds generally over all qubit states, although we are only interested in the parametric family $\eta[s,\kappa_se^{i\phi}]$.

Now let us return to our objective of proving Proposition~\ref{Scon}, which concerns the behaviour of $S(s)$. Note that $D$ is also effectively a function of $s$, given by
\be
D(s)=s(1-s)-\kappa_s^2=s(1-s)-\left|\alpha\right|^2\fr{\left[\lambda s+r-1\right]\left[\lambda r+s-1\right]}{\left[(\lambda+1)r-1\right]^2}.
\ee
Denoting the total derivative with respect to $s$ by an apostrophe, we have
\begin{align}
D'(s)&=1-2s-\left|\alpha\right|^2\fr{2\lambda s+\lambda^2r-\lambda+r-1}{\left[(\lambda+1)r-1\right]^2};\nonumber\\
D''(s)&=-2-\fr{2\lambda\left|\alpha\right|^2}{\left[(\lambda+1)r-1\right]^2}.\label{dDs}
\end{align}
The details of these derivatives are unimportant to us; what is relevant is that they are both well\hyp defined in general, as well as that $D''(s)$ is manifestly negative.

Now, exploiting the bijective relationship between $S$ and $D$, we have
\begin{align}
S'(s)&=D'(s)\fr{\rd S}{\rd D};\nonumber\\
S''(s)&=D''(s)\fr{\rd S}{\rd D}+\left[S'(s)\right]^2\fr{\rd^2S}{\rd D^2}.
\end{align}
While $D''(s)<0$ as evident from Eq.~\eqref{dDs}, Lemma~\ref{DS} establishes that $\rd S/\rd D>0$ and $\rd^2S/\rd D^2\le0$. It follows that $S''(s)\le0$, proving Proposition~\ref{Scon}.\qed

This leads immediately to
\begin{coro}
The infimum in Eq.~(\ref{defxi}) is attained for $s_j\in\{r,1-\lambda r\}$. Consequently,
\begin{align}
\xi(\bar s)&=q_{\bar s}S(r)+(1-q_{\bar s})S(1-\lambda r)\nonumber\\
&=q_{\bar s}S(\rho)+(1-q_{\bar s})h(1-\lambda r),
\end{align}
where $q_{\bar s}$ is the unique number satisfying $q_{\bar s}r+(1-q_{\bar s})(1-\lambda r)=\bar s$, namely
\be\label{defq}
q_{\bar s}=\fr{\bar s+\lambda r-1}{(1+\lambda)r-1}.
\ee
\end{coro}
Note that $q_{\bar s}$ is well\hyp defined except when $r=g$, which is anyway a trivial and uninteresting case. This brings us to our main result:
\begin{athm}[Theorem~\ref{thmain} of main text]\label{athmain}
For a qubit memory $\bM$ with Gibbs state $\gamma=\eta[g,0]$, an optimal code accessible from an initial state $\rho=\eta[r,\alpha]$ under qubit TO is
\be
\cC_{\mathrm{opt}}=\left\{\left(\fr{q_{\tilde s}}2,\eta\left[r,\pm\alpha\right]\right),\left(1-q_{\tilde s},\eta\left[1-\lambda r,0\right]\right)\right\},
\ee
where $\lambda$ is as defined in Eq.~(\ref{deflamb}) and $q_{\tilde s}$ is determined by $\tilde s$ as in Eq.~(\ref{defq}), with
\be
\tilde s:=\argmax_{\bar s\in[r,1-\lambda r]}\left[h(\bar s)-q_{\bar s}S(\rho)-(1-q_{\bar s})h(1-\lambda r)\right].
\ee
The thermal information capacity (TIC) of $\rho$ is given by the Holevo capacity of $\cC_{\mathrm{opt}}$:
\be
I(\rho)=\chi\left(\cD_{\mathrm{opt}}\right)=h(\tilde s)-q_{\tilde s}S(\rho)-(1-q_{\tilde s})h(1-\lambda r).
\ee
\end{athm}
The above optimization can easily be carried out numerically, although we have been unable to find a closed analytical form for it.

\section{The thermal information capacity in relation to other resources}\label{appscat}
Fig.~\ref{fig2d} of the main matter shows a scatter plot of the thermal information capacity (TIC) vs.\ the Gibbs free energy for a representative sample of qubit initial states. Why did we choose the Gibbs free energy as a reference against which to compare the TIC, and not other relevant measures of resourcefulness of the state? An obvious motivation for this is the asymptotic convergence of these two quantities. Nevertheless, we did also study the TIC's relation with two other resourceful aspects of the state, namely its purity (measured by the von Neumann ``negentropy'', $1-S[\rho]$) and its relative entropy of coherence with respect to the energy eigenbasis, given by
\be
C(\rho)=S\left(\rho\|\rho_\mathrm{diag}\right),
\ee
where $\rho_\mathrm{diag}$ is the diagonal part of $\rho$ in the energy eigenbasis. Like the Gibbs free energy, both the purity and the coherence are useful properties in information\hyp processing tasks, and never increase under thermal operations. In the context of the task of information storage on a memory, the purity exactly measures the information capacity in the case of an energy\hyp degenerate memory. On the other hand, states with the highest coherence, such as $\proj+$ and $\proj-$, achieve maximum TIC regardless of the temperature and energy levels. Thus, both the purity and the coherence are ostensibly indicators of the TIC. Indeed, Fig.~\ref{figPuIC} and Fig.~\ref{figCIC} bear this out. However, we see by comparing these with Fig.~\ref{fig2d} that the Gibbs free energy is the resourcefulness measure that is most strongly correlated with the TIC at all temperatures.
\begin{figure}[h]
    \subfloat[$T=0.1~\Delta E/k_{\mathrm B}$]{
    \includegraphics[width=.225\textwidth]{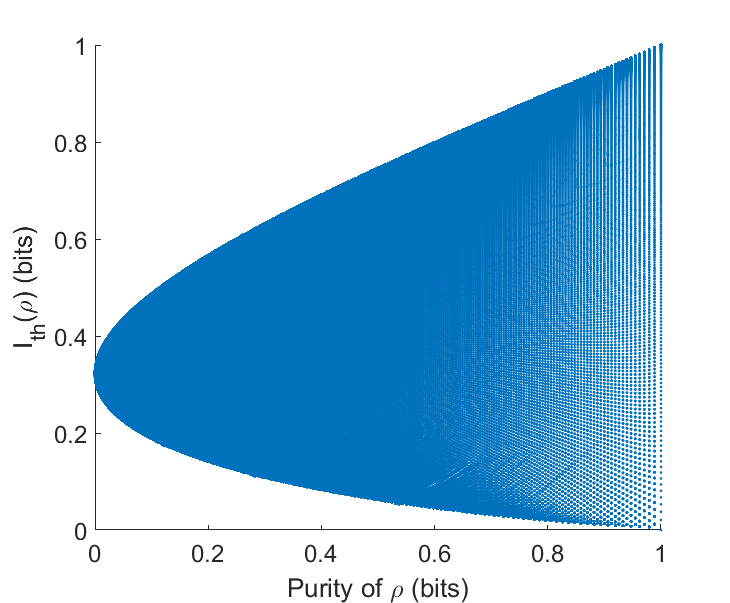}
    }
    \subfloat[$T=\Delta E/k_{\mathrm B}$]{
    \includegraphics[width=.225\textwidth]{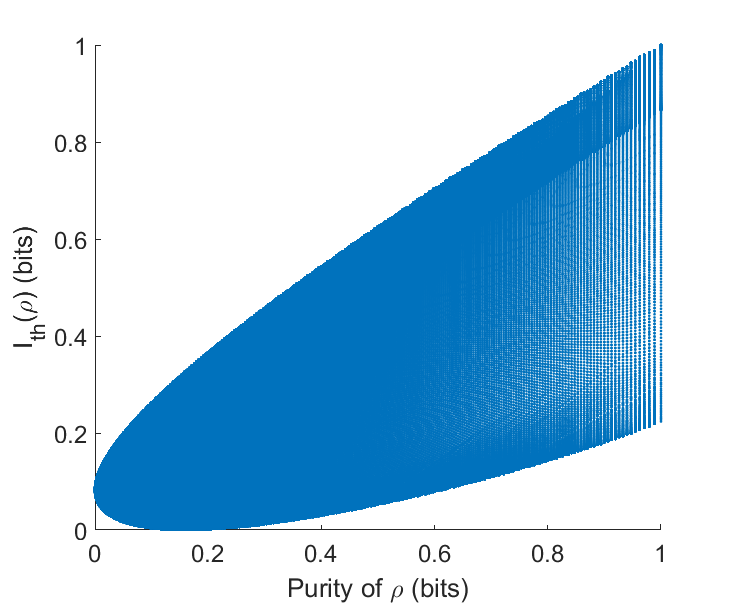}
    }
    \subfloat[$T=2~\Delta E/k_{\mathrm B}$]{
    \includegraphics[width=.225\textwidth]{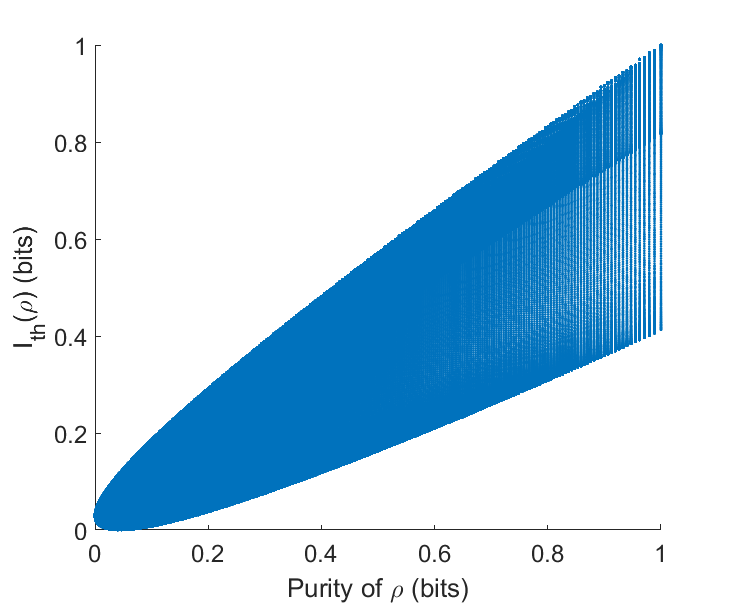}
    }
    \subfloat[$T\to\infty$]{
    \includegraphics[width=.225\textwidth]{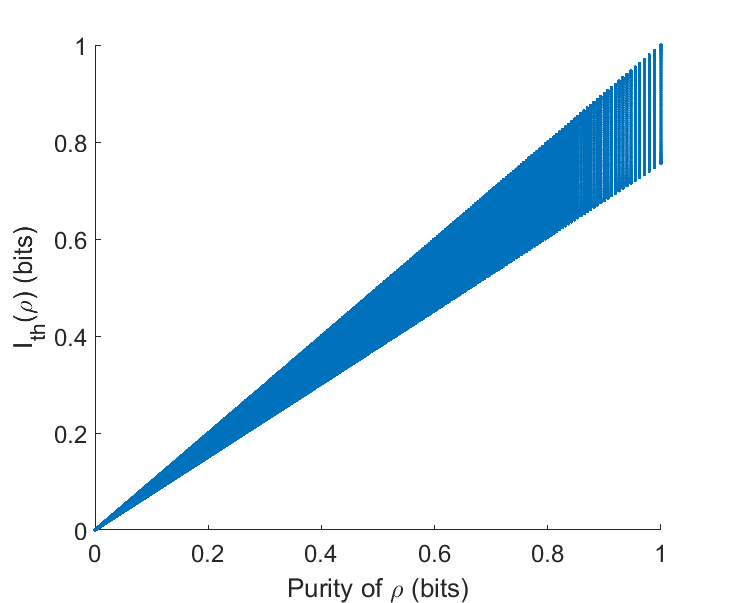}
    }
    \caption{Scatter plots of the thermal information capacity, $I_\mathrm{th}(\rho)$, vs.\ purity, $1-S(\rho)$, for qubit memory states $\rho$.}\label{figPuIC}
\end{figure}

\begin{figure}[h]
    \subfloat[$T=0.1~\Delta E/k_{\mathrm B}$]{
    \includegraphics[width=.225\textwidth]{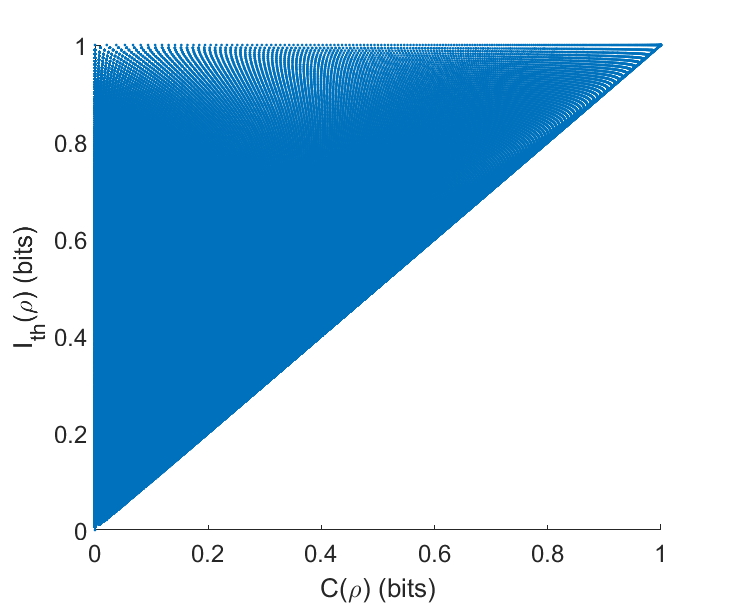}
    }
    \subfloat[$T=\Delta E/k_{\mathrm B}$]{
    \includegraphics[width=.225\textwidth]{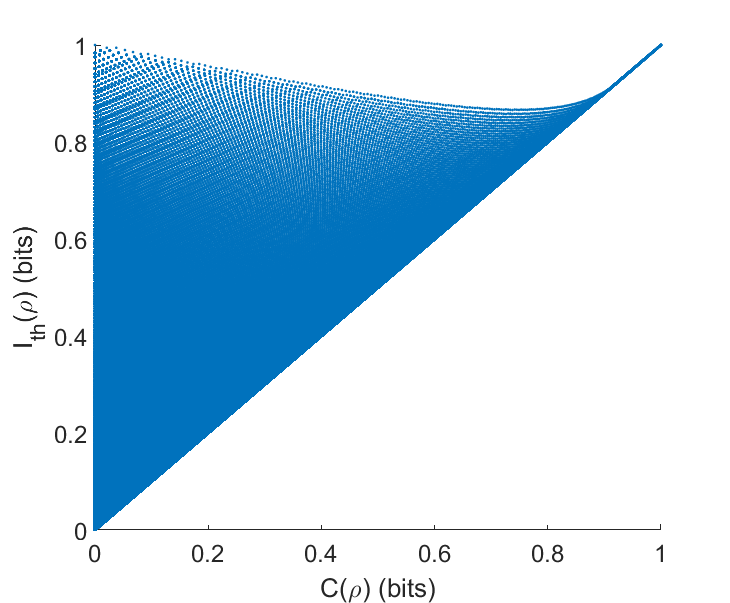}
    }
    \subfloat[$T=2~\Delta E/k_{\mathrm B}$]{
    \includegraphics[width=.225\textwidth]{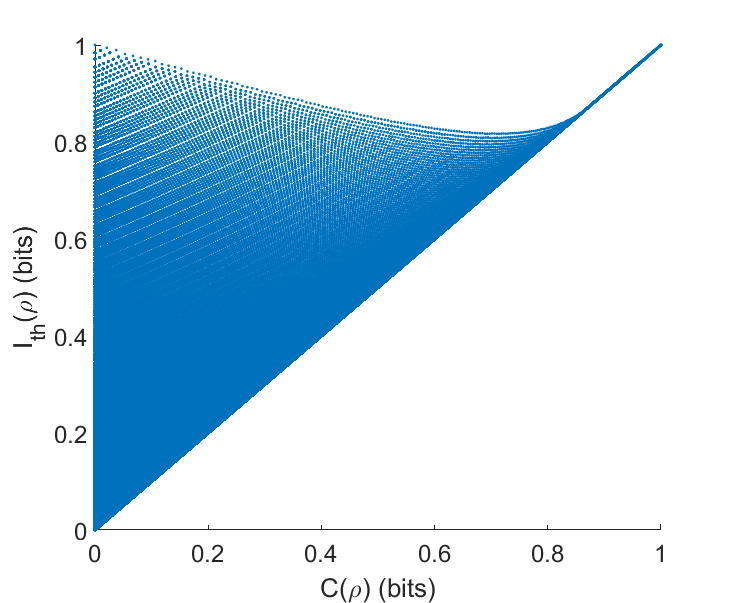}
    }
    \subfloat[$T\to\infty$]{
    \includegraphics[width=.225\textwidth]{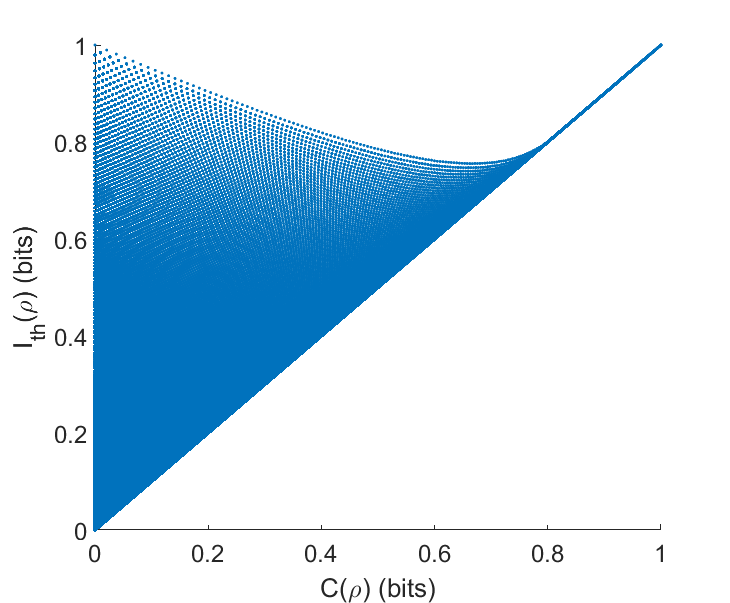}
    }
    \caption{Scatter plots of the thermal information capacity, $I_\mathrm{th}(\rho)$, vs.\ relative entropy of coherence, $C(\rho)$, for qubit memory states $\rho$.}\label{figCIC}
\end{figure}

\section{Writing on a two-level thermal memory using Jaynes\n Cummings interaction}\label{appimp}
Denote the free Hamiltonian of the memory $\bM$ by
\be
H_\bM=E_0\proj0+E_1\proj1.
\ee
Define $\omega:=\Delta E/\hbar$, and $\lambda:=\exp\left[-\Delta E/k_\rB T\right]$. In order to implement the encoding scheme of Theorem~\ref{thmain}, we need to be able to perform two transformations:
\begin{subequations}
\begin{align}
\rho\equiv r\proj0+\alpha\ket0\bra1+\alpha^*\ket1\bra0+(1-r)\proj1&\mapsto r\proj0-\alpha\ket0\bra1-\alpha^*\ket1\bra0+(1-r)\proj1;\label{trans1}\\
\rho&\mapsto(1-\lambda r)\proj0+\lambda r\proj1.\label{trans2}
\end{align}
\end{subequations}
The first, which only entails rotating the phase of the off\hyp diagonal elements of the density operator, can be achieved simply through local evolution under the memory's free Hamiltonian.

To approximate the second transformation, we propose to couple the memory to a bosonic mode bath through a Jaynes\n Cummings interaction. The bath $\bB$ is a single bosonic mode with annihilation operator $\hat b$ and free Hamiltonian
\be
H_\bB=\hbar\omega\left(\hat b^\dagger\hat b+\fr12\right).
\ee
We couple the memory with the bath through the Jaynes\n Cummings interaction term
\be
H_\rI=\Omega\left(\ket1\bra0_\bM\otimes\hat b_\bB+\ket0\bra1_\bM\otimes\hat b_\bB^\dagger\right),
\ee
where the coupling strength $\Omega\in\bbR$ can be chosen according to convenience.

Denoting the $n$\hyp boson state of $\bB$ by $\ket{b_n}$, under the global Hamiltonian
\be
H_{\bM\bB}\equiv H_\bM\otimes\eins_\bB+H_\bB\otimes\eins_\bM+H_\rI,
\ee
the higher\hyp energy components of the state undergo the well\hyp known Rabi oscillations
\be\label{Rabi}
\proj0_\bM\otimes\proj{b_{n+1}}_\bB\longleftrightarrow\proj1_\bM\otimes\proj{b_n}_\bB,
\ee
while the ground state ($\proj0_\bM\otimes\proj{b_{0}}_\bB$) component stays invariant.

The thermal state of the bath at temperature $T$ is
\be
\gamma_\bB=(1-\lambda)\sum_{n=0}^\infty\lambda^n\proj{b_n}.
\ee
For a general blank tape state $\rho_\bM\equiv\rho$ as in Eq.~\eqref{trans1}, the initial state of the composite $\bM\bB$ is the uncorrelated product $\rho_\bM\otimes\gamma_\bB$. The ground\hyp state amplitude of this state is $r(1-\lambda)$. As mentioned above, this component stays invariant, while the higher\hyp energy components oscillate within their degenerate two\hyp dimensional subspaces, as described in \eqref{Rabi}.

To achieve the transformation of \eqref{trans2}, the $\ket0_\bM$ components in all of these oscillatory terms must be transformed to $\ket1_\bM$ and vice versa. But the frequencies of the Rabi oscillations within different energy levels are not relatively rational, scaling instead as $\sqrt n$ for integer $n$, whence the desired transformations within different energy subspaces do not occur synchronously after any finite time of evolution. Nevertheless, good approximations to the desired transformation can be implemented in reasonable time. We used numerical computation to find the time taken by the Jaynes\n Cummings interaction to achieve various fractions of the optimal capacity. These ranged between $0.15$ and $1.57$ in units of $\Omega^{-1}$, where $\Omega$ is the strength of the coupling.

\end{document}